\documentclass[12pt]{article}
\usepackage{e-jc}
\usepackage{mathrsfs}

\usepackage{amsthm,amsmath,amssymb}

\usepackage{graphicx}

\usepackage[colorlinks=true,citecolor=black,linkcolor=black,urlcolor=blue]{hyperref}

\theoremstyle{plain}
\newtheorem{theorem}{Theorem}
\newtheorem{lemma}[theorem]{Lemma}
\newtheorem{corollary}[theorem]{Corollary}
\newtheorem{proposition}[theorem]{Proposition}

\theoremstyle{definition}

\newtheorem{example}[theorem]{Example}

\newtheorem{question}[theorem]{Question}

\theoremstyle{remark}

\newcommand{\pref}{\leq_p}
\newcommand{\suff}{\leq_s}

\title{\bf Infinite square-free self-shuffling words}

\author{Mike M\"uller\thanks{Supported in part by the DFG under grant 582014.}\\
\small Institut f\"ur Informatik\\[-0.8ex]
\small Christian-Albrechts-Universit{\"a}t zu Kiel\\[-0.8ex]
\small Germany\\
\small\tt mimu@informatik.uni-kiel.de\\
\and
Svetlana Puzynina \\
\small LIP, ENS de Lyon\\[-0.8ex]
\small Universit\'e
de Lyon, France\\
\small and Sobolev Institute of Mathematics \\[-0.8ex]
\small Novosibirsk, Russia\\[-0.8ex]
 \small\tt s.puzynina@gmail.com  \and
Micha\"el Rao\\
\small LIP, CNRS, ENS de Lyon\\[-0.8ex]
\small Universit\'e
de Lyon, France\\
\small\tt michael.rao@ens-lyon.fr }

\date{\dateline{Jan 1, 2012}{Jan 2, 2012}}

\begin{document}

\maketitle


\begin{abstract}
In this paper we answer two recent questions from \cite{CKPZ} and
\cite{Harju} about self-shuffling words. An infinite word $w$ is
called self-shuffling, if $w=\prod_{i=0}^\infty
U_iV_i=\prod_{i=0}^\infty U_i=\prod_{i=0}^\infty V_i$ for some
finite words $U_i$, $V_i$. Harju \cite{Harju} recently asked
whether square-free self-shuffling words exist. We answer this
question affirmatively. Besides that, we build an infinite word
such that no word in its shift orbit closure is self-shuffling,
answering positively a question from \cite{CKPZ}.

\end{abstract}

\section{Introduction}
A self-shuffling word, a notion which was recently introduced by
Charlier et al. \cite{CKPZ}, is an infinite word that can be
reproduced by shuffling it with itself. More formally, given
infinite words $x, y \in \Sigma^\omega$ over a finite alphabet
$\Sigma,$ we define ${\mathscr{S}}(x,y)\subseteq \Sigma^\omega$ to
be the collection of all infinite words $z$ for which there exists
a factorization
\[
z=\prod_{i=0}^\infty U_i V_i
\]
with each $U_i, V_i \in \Sigma^*$ and with $x=\prod_{i=0}^\infty
U_i$, $y=\prod_{i=0}^\infty V_i$. An infinite word $w\in
\Sigma^\omega$ is \emph{self-shuffling} if
$w\in{\mathscr{S}}(w,w)$. Various well-known words, e.g., the
Thue-Morse word or the Fibonacci word, were shown to be
self-shuffling.

Harju \cite{Harju} studied shuffles of both finite and infinite
square-free words, i.e., words that have no factor of the form
$uu$ for some non-empty factor $u$. More results on square-free
shuffles were obtained independently by Harju and M\"uller
\cite{HaMU13}, and Currie and Saari \cite{CuSa14}. However, the
question about the existence of an infinite square-free
self-shuffling word, posed in \cite{Harju}, remained open. We give
a positive answer to this question in Sections 2~and~3.

The \emph{shift orbit closure} $S_w$ of an infinite word $w$ can
be defined, e.g., as the set of infinite words  whose sets of
factors are contained in the set of factors of $w$. In \cite{CKPZ}
it has been proved that each word has a non-self-shuffling word
in its shift orbit closure, and the following question has been
asked: Does there exist a word for which no element of its shift
orbit closure is self-shuffling (Question 7.2)? In Section 4 we
provide a positive answer to the question.
More generally, we show the existence of a word such that for any
three words $x,y,z$ in its shift orbit closure, if $x$ is a
shuffle of $y$ and $z$, then the three words are pairwise
different. On the other hand, we show that for any infinite word
there exist three different words $x,y,z$ in its shift orbit closure such
that $x \in \mathscr{S} (y,z)$ (see Proposition~\ref{prop_xyz}).

Apart from the usual concepts in combinatorics on words, which can
be found for instance in the book of Lothaire \cite{Lothaire}, we
make use of the following notations: For every $k \geq 1$, we
denote the alphabet $\{0, 1, \ldots, k-1\}$ by $\Sigma_k$. For a
word $w = uvz$ we say that $u$ is a \emph{prefix} of $w$, $v$ is a
\emph{factor} of $w$, and $z$ is a \emph{suffix} of $w$. We denote
these prefix- and suffix relations by $u \pref w$ and $v \suff w$,
respectively.
By $w[i,j]$ we denote the factor of $w$ starting at position $i$ and
ending after position $j$.
Note that we start numbering the positions with $0$.

A \emph{prefix code} is a set of words with the property that none
of its elements is a prefix of another element. Similarly, a
\emph{suffix code} is a set of words where no element is a suffix
of another one. A \emph{bifix code} is a set that is both a prefix
code and a suffix code. 
A morphism $h$ is \emph{square-free} if for all square-free words
$w$, the image $h(w)$ is square-free.

\section{A square-free self-shuffling word on four letters}
Let $g : \Sigma_4^* \rightarrow \Sigma_4^*$ be the morphism
defined as follows:
\begin{align*}
    g(0) &= 0121, \\
    g(1) &= 032, \\
    g(2) &= 013, \\
    g(3) &= 0302.
\end{align*}

We will show that the fixed point $w = g^\omega(0)$ is square-free
and self-shuffling in the following. Note that $g$ is not a
square-free morphism, that is, it does not preserve
square-freeness, as $g(23) = 0130302$ contains the square $3030$.

\begin{lemma}\label{no_3u1u3}
    The word $w = g^\omega(0)$ contains
    no factor of the form $3u1u3$ for some $u \in \Sigma_4^*$.
\end{lemma}
\begin{proof}
    We assume that there exists a factor of the form $3u1u3$ in $w$, for some word $u \in \Sigma_4^*$.
    From the definition of $g$, we observe that $u$ can not be empty.
    Furthermore, we see that every $3$ in $w$ is preceded by either $0$ or $1$.
    If $1 \suff u$, then we had an occurrence of the factor $11$ in $w$, which is not possible by the definition of $g$, hence $0 \suff u$.
    Now, every $3$ is followed by either $0$ or $2$ in $w$ and $01$ is followed by either $2$ or $3$.
    Since both $3u$ and $01u$ are factors of $w$, we must have $2 \pref u$.
    This means that the factor $012$ appears at the center of $u1u$, which can only be followed by $1$ in $w$, thus $21 \pref u$.
    However, this results in the factor $321$ as a prefix of $3u1u3$, which does not appear in $w$, as seen from the definition of $g$.
\end{proof}

\begin{lemma}
    The word $w = g^\omega(0)$ is square-free.
\end{lemma}
\begin{proof}
    We first observe that $\{g(0), g(1), g(2), g(3)\}$ is a bifix code.
    Furthermore, we can verify that there are no squares $uu$ with $|u| \leq 3$ in $w$.
    Let us assume now, that the square $uu$ appears in $w$ and that $u$ is the shortest word with this property.
    If $u = 02u'$, then $u' = u''03$ must hold, since $02$ appears only as a factor of $g(3)$, and thus $uu$ is a suffix of the factor $g(3)u''g(3)u''$ in $w$.
    As $w = g(w)$, also the shorter square $3g^{-1}(u'')3g^{-1}(u'')$ appears in $w$, a contradiction.
    The same desubstitution principle also leads to occurrences of shorter squares in $w$ if $u = xu'$ and $x \in \{01, 03, 10, 12, 13, 21, 30, 32\}$.

    If $u = 2u'$ then either $03 \suff u$ or $030 \suff u$ or $01 \suff u$, by the definition of $g$.
    In the last case, that is when $01 \suff u$, we must have $21 \pref u$, which is covered by the previous paragraph.
    If $u' = u''030$, then $uu$ is followed by $2$ in $w$ and we can desubstitute to obtain the shorter square $g^{-1}(u'')3g^{-1}(u'')3$ in $w$.
    If $u = 2u'$ and $u' = u''03$, and $uu$ is preceded by $03$ or followed by $2$ in $w$, we can desubstitute to $1g^{-1}(u'')1g^{-1}(u'')$ or $g^{-1}(u'')1g^{-1}(u'')1$, respectively.
    Therefore, assume that $u = 2u''03$ and $uu$ is preceded by $030$ and followed by $02$ in $w$.
    This however means that we can desubstitute to get an occurrence of the factor $3g^{-1}(u'')1g^{-1}(u'')3$ in $w$, a contradiction to Lemma \ref{no_3u1u3}.
\end{proof}

We now show that $w = g^\omega(0)$ can be written as
$w=\prod_{i=0}^\infty U_iV_i=\prod_{i=0}^\infty
U_i=\prod_{i=0}^\infty V_i$ with $U_i,V_i \in \Sigma_4^*$.

\begin{lemma}
    The word $w = g^\omega(0)$ is self-shuffling.
\end{lemma}
\begin{proof}
    In what follows we use the notation $x=v^{-1}u$ meaning that $u=v
    x$ for finite words $x, u, v$. We are going to show that the
    self-shuffle is given by the following:

    \small\begin{align*}& U_0=g^2(0),& &\hskip-9pt U_1=0,& &\hskip-10pt \dots,\hskip-8pt & &U_{6i+2}=g^i(0^{-1}g(0)0),&
    &\hskip-6pt U_{6i+3}=g^{i}(0^{-1}g(3)0),\\& & & & & &
    &U_{6i+4}=g^{i}(0^{-1}g(201)0),& & \hskip-6pt U_{6i+5}=g^{i}(30),
    \\& & & & & &
    &U_{6i+6}=g^{i}(2g(03)),& &
    \hskip-6pt U_{6i+7}=g^{i+1}(20), \\
    & V_0=g(0)03,& &\hskip-9pt V_1=2g(2)0,& &\hskip-10pt \dots,
    \hskip-8pt & &V_{6i+2}=g^{i}(0^{-1}g(1)0),& & \hskip-6pt
    V_{6i+3}=g^{i}(0^{-1}g(03)0),\\& & & & & & &V_{6i+4}=g^{i}(1),& &
    \hskip-6pt V_{6i+5}=g^{i}(3),\\& & & & & & &V_{6i+6}=g^{i+1}(0),&
    & \hskip-6pt V_{6i+7}=g^{i+1}(0^{-1}g(2)0).
    \end{align*}\normalsize

    Now we verify  that $$w=\prod_{i=0}^\infty U_i V_i
    =\prod_{i=0}^\infty U_i =\prod_{i=0}^\infty V_i,$$ from which it
    follows that $w$ is self-shuffling. It suffices to show that each
    of the above products is fixed by $g$. Indeed, straightforward
    computations show that

    $$\prod_{i=0}^\infty U_i = g^2(0) g^2(121) g^3(121) \cdots \, , $$

    which is fixed by $g$:

    \begin{align*} g\left(\prod_{i=0}^\infty U_i\right) &= g\left(g^2(0) g^2(121) g^3(121) \cdots \right) = g^3(0) g^3(121) g^4(121) \cdots \\
                                             &= g^2(0121) g^3(121) g^4 (121) \cdots = g^2(0) g^2(121) g^3(121) \cdots =\prod_{i=0}^\infty U_i, \end{align*}
    hence $\prod_{i=0}^\infty U_i$ is fixed by $g$ and thus
    $w=\prod_{i=0}^\infty U_i$. In a similar way we show that
    $w=\prod_{i=0}^\infty V_i=\prod_{i=0}^\infty U_i V_i$.
\end{proof}

\section{Square-free self-shuffling words on three letters}

We remark that we can immediately produce a square-free
self-shuffling word over $\Sigma_3$ from $g^\omega(0)$: Charlier
et al. \cite{CKPZ} noticed that the property of being
self-shuffling is preserved by the application of a morphism.
Furthermore, Brandenburg \cite{Brandenburg} showed that the
morphism $f : \Sigma_4^* \rightarrow \Sigma_3^*$, defined by
\begin{align*}
    f(0) &= 010201202101210212, \\
    f(1) &= 010201202102010212, \\
    f(2) &= 010201202120121012, \\
    f(3) &= 010201210201021012,
\end{align*}
is square-free. Therefore, the word $f(g^\omega(0))$ is a ternary
square-free self-shuffling word, from which we can produce a
multitude of others by applying square-free morphisms from
$\Sigma_3^*$ to $\Sigma_3^*$.

\section{A word with non self-shuffling shift orbit closure}
In this section we provide a positive answer to the question from
\cite{CKPZ} whether there exists a word for which no element of
its shift orbit closure is self-shuffling.

The \emph{Hall word} $\mathcal{H}=012021012102\cdots$ is defined
as the fixed point of the morphism $h(0) =012, h(1)= 02, h(2) =1$.
Sometimes it is referred to as a \emph{ternary Thue-Morse} word.
It is well known that this word is square-free. We show that no
word in the shift orbit closure $S_\mathcal{H}$ of the Hall word
is self-shuffling. More generally, we show that if  $x$ is a
shuffle of $y$ and $z$ for $x,y,z\in S_\mathcal{H}$, then they are
pairwise different.

\begin{proposition}\label{prop_xyy} There are no words $x, y$ in the shift orbit
closure of the Hall word such that $x \in \mathscr{S}
(y,y)$.\end{proposition}

\begin{proof}
Suppose the converse, i.e., there exist words $x, y\in
S_\mathcal{H}$ such that

$$ x = \prod_{i=0}^{\infty} U_i V_i, \qquad y= \prod_{i=0}^{\infty} U_i = \prod_{i=0}^{\infty} V_i.$$

Define the set $X$ of infinite words as follows:

$$X=\{x \in S_\mathcal{H} \, \mid \, x \in \mathscr{S} (y,y) \mbox{ for some } y\in S_\mathcal{H} \}.$$
In other words, $X$ consists of words in $S_\mathcal{H}$ which can
be introduced as a shuffle of some word $y$ in $S_\mathcal{H}$
with itself.
Now suppose, for the sake of contradiction, that $X$ is non empty, and
consider $x\in X$ with the first block $U_0$ of the smallest
possible positive length. We remark that such $x$ and
corresponding $y$ are not necessarily unique.

We can suppose without loss of generality that $y$ starts with $0$
or $10$. Otherwise, we exchange $0$ and $2$, consider the
morphism $0 \mapsto 1, 1 \mapsto 20, 2\mapsto 210$, and the
argument is symmetric.

It is not hard to see from the properties of the morphism $h$ that
removing every occurrence of $1$ from $y$ results in $(02)^\omega$. Hence the
blocks in the factorizations of $y$ after removal of $1$ are of the form
$(02)^i$ for some integer $i$. Thus the first letter of each block $U_i$ and
$V_i$ that is different from $1$ is $0$, and the last letter different from $1$ is $2$.

Then, $U_i$ and $V_i$ are images by the morphism $h$ of factors of
the fixed point of $h$. Therefore, there are words $x', y'\in
S_{\mathcal{H}}$ such that $x=h(x'), y=h(y'), U_i=h(U'_i),
V_i=h(V'_i),$ and $x'= \prod_{i=0}^\infty U'_i V'_i$, $y' = \prod_{i=0}^\infty U'_i =
\prod_{i=0}^\infty V'_i $.

Notice that the first block $U_0$ cannot be equal to $1$. Indeed,
otherwise $x$ starts with $11$, which is impossible, since $11$ is
not a factor of the fixed point of $h$.

Clearly, taking the preimage decreases the lengths of blocks in
the factorization (except for those equal to $1$), and since
$U_0\neq 1$, the length of the first block in the preimage is
smaller, i.e., $|U_0'|<|U_0|$. This is a contradiction with the minimality of
$|U_0|$.
\end{proof}

\begin{corollary} There are no self-shuffling words in the shift orbit
closure of $\mathcal{H}$. \end{corollary}

With a similar argument we can prove the following:

\begin{proposition}\label{prop_xxy} There are no words $x, y$ in the shift orbit
closure of $\mathcal{H}$ such that $x \in \mathscr{S} (x,y)$.
\end{proposition}

\begin{proof} First we introduce a notation $x\in \mathscr{S}_2
(y,z)$, meaning that there exists a shuffle starting with the word
$z$ (i.e., $U_0=\varepsilon$, $V_0 \neq \varepsilon$). Next, $x \in
\mathscr{S} (x,y)$ implies that there exists $z$ in the same shift
orbit closure such that $z \in \mathscr{S}_2 (z,y)$. Indeed, one
can remove the prefix $U_0$ of $x$ to get $z$, i.e., $z=(U_0)^{-1}x$, and
keep all the other blocks $U_i$, $V_i$ in the shuffle product.

Define the set $Z$ of infinite words as follows:

$$Z=\{z \in S_{\mathcal{H}} \, \mid \, z \in \mathscr{S}_2 (z,y) \mbox{ for some } y\in S_\mathcal{H} \}.$$
In other words, $Z$ consists of words in $S_\mathcal{H}$ which can
be introduced as a shuffle of some word $y$ in $S_\mathcal{H}$
with $z$ starting with the block $V_0$. Now consider $z\in Z$ with
the first block $V_0$ of the smallest possible length. We remark
that such $z$ and a corresponding $y$ are not necessarily unique.

As in the proof of Proposition \ref{prop_xyy}, the shuffle cannot
start with a block of length $1$. Again, if we remove every occurrence of $1$ in
$y$ (and in $z$), we get $(02)^\omega$ or $(20)^\omega$; moreover,
since $V_0$ contains letters different from $1$, the first letter different from $1$ is the same in $y$ and $z$.
So, without loss of generality we assume that
both $y$ and $z$ without $1$ are $(02)^\omega$, and the blocks
$U_i$ and $V_i$  without $1$ are integer powers of $02$. Then,
$U_i$ and $V_i$ are images by the morphism $h$ of factors of
$\mathcal{H}$. Therefore, there are words $z', y'\in
S_\mathcal{H}$ such that $z=h(z'), y=h(y'), U_i=h(U'_i),
V_i=h(V'_i),$ and $z'= \prod_{i=0}^\infty (U'_i V'_i) = \prod_{i=0}^\infty V'_i$, $y' = \prod_{i=0}^\infty
U'_i $ (i.e., $z'\in Z$).

As in the proof of Proposition \ref{prop_xyy}, since $V_0\neq 1$,
the length of the first block in the preimage is smaller, i.e.,
$|V_0'|<|V_0|$. This is again a contradiction with the minimality of $|V_0|$.
\end{proof}

So, we proved that if there are three words $x, y, z$ in the shift
orbit closure of the fixed point of $h$ such that $x \in
\mathscr{S} (y,z)$, then they should be pairwise distinct. Now we
are going to prove that for any infinite word there exist three
different words in its shift orbit closure such
that $x \in \mathscr{S} (y,z)$.

An infinite word $x$ is called \emph{recurrent}, if each its
prefix occurs infinitely many times in~it.

\begin{proposition}\label{prop_xyz} Let $x$ be a recurrent infinite word. Then there exist two words $y, z$ in the shift orbit
closure of $x$ such that $x \in \mathscr{S}
(y,z)$.\end{proposition}
\begin{proof} We build the shuffle inductively.

Start from any prefix $U_0$ of $x$. Since $x$ is recurrent, each
of its prefixes occurs infinitely many times in it. Find another
occurrence of $U_0$ in $x$ and denote its position by $i_1$. Put
$V_0= x[|U_0|, i_1+|U_0|-1]$.


At step $k$, suppose that the shuffle of the prefix of $x$ is
built:
$$x[0, \Sigma_{l=0}^{k-1}(|U_l|+|V_l|)-1]=\prod_{l=0}^{k-1} U_l V_l, \quad
y[0, \Sigma_{l=0}^{k-1}|U_l|-1]=\prod_{l=0}^{k-1} U_l, \quad z[0,
\Sigma_{l=0}^{k-1} |V_l|-1]=\prod_{l=0}^{k-1} V_l, $$ such that
$\prod_{l=0}^{k-1} U_l$ is the suffix of $x[0,
\Sigma_{l=0}^{k-1}(|U_l|+|V_l|)-1]=\prod_{i=0}^{k-1} U_l V_l$
starting at position $i_{k-1}$, and  
$\prod_{l=0}^{k-1} V_l$ is the suffix of $x[0,
\Sigma_{l=0}^{k-1}(|U_l|+|V_l|)-1]=\prod_{i=0}^{k-1} U_l V_l$
starting at position $j_{k-1}$.

Find another occurrence of $\prod_{l=0}^{k-1} V_l$ 
in $x$ at some
position $j_k > j_{k-1}$. We can do it since $x$ is recurrent. Put
$U_k= x[\Sigma_{l=0}^{k-1}(|U_l|+|V_l|), 
j_k-1+\Sigma_{l=0}^{k-1}|V_l|]$. We note that 
$\prod_{l=0}^{k} U_l$ is a
factor of $x$ by the construction; more precisely, it occurs at
position $i_{k-1}$.

Find an occurrence of $\prod_{l=0}^{k} U_l$ at some position
$i_k>i_{k-1}$, put $V_k= x[\Sigma_{l=0}^{k-1}(|U_l|+|V_l|) +
|U_k|, i_k-1+\Sigma_{l=0}^{k}|U_l|]$. As above, $\prod_{l=0}^{k} V_l$
is a factor of $x$ by the construction since it occurs at
position $j_{k-1}$. Moreover, both $\prod_{l=0}^{k} U_l$  and $\prod_{l=0}^{k} V_l$ 
are suffixes of $x[0,
\Sigma_{l=0}^{k}(|U_l|+|V_l|)-1]=\prod_{i=0}^{k} U_l V_l$.

Continuing this line of reasoning, we build the required
factorization.
\end{proof}

Since each infinite word contains a recurrent (actually, even a
uniformly recurrent) word in its shift orbit closure, we obtain
the following corollary:
\begin{corollary}\label{col:xyz} Each infinite word $w$ contains words $x, y, z$ in its shift orbit
closure such that $x \in \mathscr{S} (y,z)$.
\end{corollary}

The following example shows that the recurrence condition in
Proposition \ref{prop_xyz} cannot be omitted:

\begin{example} Consider the word $3\mathcal{H} = 3012021\cdots$ which
is obtained from $\mathcal{H}$ by adding a letter $3$ in the
beginning. Then the shift orbit closure of $3\mathcal{H}$ consists
of the shift orbit closure of $\mathcal{H}$ and the word
$3\mathcal{H}$ itself. Assuming $3\mathcal{H}$ is a shuffle of two words in
its shift orbit closure, one of them is $3\mathcal{H}$ (there are
no other $3$'s) and the other one is something in the shift orbit
closure of $\mathcal{H}$, we let $y$ denote this other word.
Clearly, the shuffle starts with $3\mathcal{H}$, and cutting the
first letter $3$, we get $\mathcal{H}\in
\mathscr{S}(\mathcal{H},y)$, a contradiction with Proposition~\ref{prop_xxy}.
\end{example}

There also exist examples where each letter occurs infinitely many
times:

\begin{example} The following word:
$$x=012001120001112\cdots0^k1^k2\cdots$$ does not have
two words $y,z$ in its shift orbit closure such that $x \in
\mathscr{S} (y,z)$. The idea of the proof is that the shift orbit
closure consists of words of the following form: $1^*20^{\omega}$,
$0^*1^{\omega}$, $x$ itself and all their right shifts. Shuffling
any two words of those types, it is not hard to see that there
exists a prefix of the shuffle which contains too many or too few
occurrences of some letter compare to the prefix of $x$. We leave
the details of the proof to the reader.
\end{example}

By Corollary~\ref{col:xyz}, there are $x,y,z$ in the shift
orbit closure of $\mathcal{H}$ such that $x\in \mathscr{S}(y,z)$.
To conclude this section, we give an explicit construction of two
words in the shift orbit closure of $\mathcal{H}$ which can be
shuffled to give $\mathcal{H}$. We remark though that this
construction gives a shuffle different from the one given by
Corollary~\ref{col:xyz}. Let:
$$h:\left\{
\begin{array}{lll}
0&\mapsto&012\\
1&\mapsto&02\\
2&\mapsto&1\\
\end{array}\right. \text{ and ~ }
h':\left\{
\begin{array}{lll}
0&\mapsto&210\\
1&\mapsto&20\\
2&\mapsto&1.\\
\end{array}\right.$$

By definition, the shift orbit closure of the Hall word is closed
under $h$. Moreover this shift orbit closure is also closed under $h'$,
since the factors of the Hall word are closed under the morphism
$0\to 2, 1\to 1, 2\to 0$.
\small
$$
h'\circ h:\left\{
\begin{array}{lll}
0&\mapsto&210201\\
1&\mapsto&2101\\
2&\mapsto&20\\
\end{array}\right.
h\circ h':\left\{
\begin{array}{lll}
0&\mapsto&102012\\
1&\mapsto&1012\\
2&\mapsto&02\\
\end{array}\right.
h^2:\left\{
\begin{array}{lll}
0&\mapsto&012021\\
1&\mapsto&0121\\
2&\mapsto&02\\
\end{array}\right.
h'^2:\left\{
\begin{array}{lll}
0&\mapsto&120210\\
1&\mapsto&1210\\
2&\mapsto&20.\\
\end{array}\right.$$
\normalsize Note that if $w$ is an infinite word, then $2 (h\circ
h')(w) = (h'\circ h)(w)$ and $0h'^2(w)=h^2(w)$.

\def\prodi{\prod_{i=0}^\infty}


\begin{theorem}
$h^\omega(0) \in \mathscr{S}
(h^2((h'^2)^\omega(1)),h'^3(h^\omega(0)))$.
\end{theorem}
\begin{proof}
Let 
$$U_0=01, \, U_1=h'(0), \, U_2=h'(1), \, V_0=h'(1),$$
and for every $i\ge 0$,
$$U_{i+3} = h'^2(h^i(1)) \text{ and } V_{i+1}=h'^2(h^i(1)).$$
Let furthermore
$$u=\prod_{i=0}^\infty U_i, \, v=\prod_{i=0}^\infty V_i, \text{ and } w=\prod_{i=0}^\infty U_i V_i \, .$$
We show that $w= h^\omega(0)$, $u=h^2((h'^2)^\omega(1))$ and
$v=h'^3(h^\omega(0))$.

Note that $2h(h'(h^\omega(0)))= h'(h^\omega(0))$, thus
$h'(h^\omega(0))=\prodi h^i(2)$. Then we have $$v=20
\prod_{i=0}^\infty h'^2(h^i(1)) = h'^2 \left (\prod_{i=0}^\infty
h^i(2)\right )= h'^3(h^\omega(0)).$$ Moreover,
\begin{align*} 
    u  &= 01 210 20 \prod_{i=0}^\infty h'^2(h^i(1))= 01210h'^2\left (\prodi h^i(2)\right ) \\
       &= 01210h'^3\left(h^\omega(0)\right)= 01h'(0h'^2(h^\omega(0)))=
01h'(h^\omega(0))=0 h'(2h^\omega(0)).
\end{align*} 
Since
$h'^2(2h^\omega(0)) = 20 h'^2(h^\omega(0))= 2 h^\omega(0)$, 
the word $2h^\omega(0)$ is the fixed point
$(h'^2)^\omega(2)$ of $h'^2$, and then $h'(2h^\omega(0))$ is the fixed point
$(h'^2)^\omega(1)$. Thus $u  = 0 (h'^2)^\omega(1) =
h^2((h'^2)^\omega(1))$. Finally:
$$w=0120210121020\prod_{i=0}^\infty h'^2(h^i(021))
    =012 021 h(021) h^2\left (\prod_{i=0}^\infty h^i(021) \right)
=     012 \prod_{i=0}^\infty h^i(021).$$

Applying the morphism $h$ to the second expression for $w$, we get
$$h(w)=     012021 h\left (\prod_{i=0}^\infty h^i(021) \right ) = 012
\prod_{i=0}^\infty h^i(021).$$


Thus $w= h^\omega(0)$ since $h$ is injective.
\end{proof}

\section{Conclusion and open question}
We showed that infinite square-free self-shuffling words exist.
The natural question that arises now is whether we can find
infinite self-shuffling words subject to even stronger
avoidability constraints: For this we recall the notion of
\emph{repetition threshold} $RT(k)$, which is defined as the least
real number such that an infinite word over $\Sigma_k$
exists, that does not contain repetitions of exponent greater than
$RT(k)$.
Due to the collective effort of many researchers
~(see \cite{CurrieR11,Rao11} and references therein), the
repetition threshold for all alphabet sizes is known and
characterized as follows:
$$RT(k) =
  \begin{cases}
   \frac{7}{4}       & \text{if } k=3 \\
   \frac{7}{5} & \text{if } k=4 \\
   \frac{k}{k-1} & \text{else}.
  \end{cases}
$$
A word $w \in \Sigma_k^\omega$ without factors of exponent greater
than $RT(k)$ is called a \emph{Dejean word}. Charlier et al.
showed that the Thue-Morse word, which is a binary Dejean word, is
self-shuffling \cite{CKPZ}.

\begin{question} Do there exist self-shuffling Dejean words over
non-binary alphabets? \end{question}

\bibliographystyle{abbrv}

\end{document}